\documentclass[11pt,twoside]{article}
\usepackage{informat}
\newtheorem{definition}{Definition}
\newtheorem{theorem}{Theorem}
\newtheorem{proposition}{Proposition}
\newtheorem{corollary}{Corollary}
\newcommand{\spk}[1]{sp\langle #1 \rangle}
\newcommand{\spfd}{\rightarrow_{sp}}
\newcommand{\nul}{\texttt{NULL}}
\newenvironment{proof}{\textsc{Proof:}}{\hfill$\Box$}
\usepackage{Informat,amssymb,amsmath}
\begin{document}

\title{Approximate Keys and Functional Dependencies in Incomplete Databases With Limited Domains--Algorithmic Perspective}
\author{
Munqath Alatar\\
ITRDC, University of Kufa, Iraq\\
\texttt{munqith.alattar@uokufa.edu.iq}\\
\AndAuthor
Attila Sali\\
Alfr\'ed R\'enyi Institute of Mathematics\\
and Department of Computer Science,\\
Budapest University of Technology and Economics\\
Budapest, Hungary\\
       \texttt{sali.attila@renyi.hu}
}
  \titleodd{Approximate Keys and Functional Dependencies\ldots}
  \authoreven{M. Alatar et al.}
  \received{February 1, 2023}

  \abstract{A possible world of an incomplete database table is obtained by imputing values from the attributes (infinite) domain to the place of \texttt{NULL} s.  A table satisfies a possible key or possible functional dependency constraint if there exists a possible world of the table that satisfies the given key or functional dependency constraint. A certain key or functional dependency is satisfied by a table if all of its possible worlds satisfy the constraint. Recently, an intermediate concept was introduced. A strongly possible key or functional dependency is satisfied by a table if there exists a strongly possible world that satisfies the key or functional dependency. A strongly possible world is obtained by imputing values from the active domain of the attributes, that is from the values appearing in the table. In the present paper, we study approximation measures of strongly possible keys and FDs. Measure $g_3$ is the ratio of the minimum number of tuples to be removed in order that the remaining table satisfies the constraint. We introduce a new measure $g_5$, the ratio of the minimum number of tuples to be added to the table so the result satisfies the constraint. $g_5$ is meaningful because the addition of tuples may extend the active domains. We prove that if $g_5$ can be defined for a table and a constraint, then the $g_3$ value is always an upper bound of the $g_5$ value. However, the two measures are independent of each other in the sense that for any rational number $0\le\frac{p}{q}<1$ there are tables of an arbitrarily large number of rows and a constant number of columns that satisfy $g_3-g_5=\frac{p}{q}$. A possible world is obtained usually by adding many new values not occurring in the table before. The measure $g_5$ measures the smallest possible distortion of the active domains.
  We study complexity of determining these approximate measures.
}
\maketitle
\section{Introduction}
The information in many industrial and research databases may usually be incomplete due to many reasons. 
For example, databases related to instrument maintenance, medical applications, and surveys \cite{farhangfar2007novel}.
This makes it necessary to handle the cases when some information missing from a database and are required by the user. Imputation (filling in) is one of the common ways to handle the missing values \cite{lipski1981databases}. 

A new approach for imputing values in place of the missing information was introduced in \cite{alattar2019keys}, to achieve complete data tables, using only information already contained in the SQL table attributes (which are called the active domain of an attribute). Any total table obtained in this way is called a \emph{strongly possible world}. We use only the data shown on the table to replace the missing information because in many cases there is no proper reason to consider any other attribute values than the ones that  already exist in the table.
Using this concept, new key and functional dependency constraints called strongly possible keys (spKeys) and strongly possible functional dependencies (spFDs) were defined in \cite{alattar2020strongly,alattar2020functional}   that are satisfied after replacing any missing value (\texttt{NULL}) with a value that is already shown in the corresponding attribute. 
In Section~\ref{defs}, we provide the formal definitions of spKeys and spFDs.

The present paper continues the work started in  \cite{alattar2020strongly}, where an approximation notion was introduced to calculate how close any given set of attributes can be considered as a key. Tuple removal may be necessary because the active domains do not contain enough values to be able to replace the \texttt{NULL}\ values so that the tuples are pairwise distinct on a candidate key set of attributes $K$. In the present paper, we study approximation measures of spKeys and spFDs by adding tuples. Adding a tuple with new unique values will add more values to the attributes' active domains, thus some unsatisfied constraints may get satisfied.

For example,  $Car\_Model$ and $Door No$ is designed to form a key in the Cars Types table shown in Table \ref{add_vs_rmv} but the table does not satisfy the spKey $\spk{Car\_Model, Door No}$.  Two tuples would need to be removed, but  adding a new tuple with distinct door number value to satisfy $\spk{Car\_Model, Door No}$ is better than removing two tuples. 
In addition to that, we know that the car model and door number determines the engine type, then the added tuple can also have a new value in the $Door No$ attribute so that the table satisfy $(Car\_Model, Door No) \spfd Engine\_Type$ rather than removing other two tuples.
\begin{table}[ht]\caption{Cars Types Incomplete Table}\label{add_vs_rmv}
    \centering
\begin{tabular}{ c c c }
		\hline
		\textbf{Car\_Model} & \textbf{Door No} & \textbf{Engine\_Type}  \\  \hline
		 BMW I3 & 4 doors & $\bot$ \\ 
		 BMW I3 & $\bot$ & electric \\
		 Ford explorer & $\bot$ & V8\\
		 Ford explorer & $\bot$ & V6\\
		 \hline
	\end{tabular}
\end{table}

\section{Definitions}\label{defs}
Let $ R = \{ A_{1},A_{2},\ldots A_{n}\} $ be a relation schema. The set of all the possible values for each attribute $ A_i \in R $ is called the domain of $A_i$ and denoted as $ D_{i}$ = $dom(A_{i})$ for $i$ = 1,2,$\ldots n$. Then, for $X\subseteq R$, then $D_X = \prod\limits_{\forall A_i \in K} D_i$.

An instance $T$ = ($t_{1}$,$t_{2}, \ldots t_{s}$) over $R$ is a list of tuples such that each tuple is a function $t :  R \rightarrow \bigcup_{A_i\in R} dom(A_i)$ and $t[A_i] \in dom(A_i)$ for all $A_i$ in $R$. By taking a list of tuples we use the \emph{bag semantics} that allows several occurrences of the same tuple. Usage of the bag semantics is justified by that SQL allows multiple occurrences of tuples. Of course, the order of the tuples in an instance is irrelevant, so mathematically speaking we consider a \emph{multiset of tuples} as an instance.
For a tuple $t_{r} \in T$ and $X \subset R$, let $t_{r}[X]$ be the restriction of $t_{r}$  to $X$.

It is assumed that $\bot$ is an element of each attribute's domain that denotes missing information.
$t_{r}$ is called $V$-total for a set $V$ of attributes  if $\forall A\in V$, $t_r[A]\neq\bot$.  Also, $t_{r}$ is a total tuple if it is $R$-total.   $t_{1}$ and $t_{2}$ are \emph{weakly similar} on $X \subseteq R$ denoted as $t_{1}[X] \sim_{w} t_{2}[X]$ defined by K\"ohler et.al. \cite{kohler2016possible} if 

\[ \forall A \in X \quad (t_{1}[A] = t_{2}[A] \textrm{ or } t_{1}[A] =\bot \textrm{ or } t_{2}[A] = \bot). \]

Furthermore, $t_{1}$ and $t_{2}$ are \emph{strongly similar} on $X \subseteq R$ denoted by $t_{1}[X] \sim_{s} t_{2}[X]$ if

$$  \forall A \in X \quad  (t_{1}[A] = t_{2}[A] \neq \bot).$$
For the sake of convenience we write $t_{1} \sim_{w} t_{2}$ if $t_{1}$ and $t_{2}$ are weakly similar on $R$ and use the same convenience for strong similarity.
Let $T= (t_{1},t_{2}, \ldots t_{s})$ be a table instance over $R$. Then, $T'= (t'_{1},t'_{2}, \ldots t'_{s})$ is a \emph{possible world} of $T$, if $t_{i} \sim_{w} t'_{i}$ for all $i=1,2,\ldots s$ and $T'$ is completely \texttt{NULL} -free. That is, we replace the occurrences of $\bot$ with a value from the domain $ D_{i} $ different from $\bot$ for all tuples and all attributes. Active domain of an attribute is the set of all the distinct values shown under the attribute except the \texttt{NULL}. Note that this was called the \emph{visible domain} of the attribute in papers \cite{alattar2019keys,alattar2020functional,alattar2020strongly,al2022strongly}. 

\begin{definition}\label{vd-def} The \emph{active domain} of an attribute $ A_i $ ($ VD^T_{i} $) is the set of all distinct values except $\bot$ that are already used by tuples in $ T $:
$$ VD^T_{i} = \{t[A_{i}] : t \in T\} \setminus \{\bot \} \textrm{ for } A_i \in R.$$ 
\end{definition}To simplify notation, we omit the upper index $T$ if it is clear from the context what instance is considered.

While a possible world is obtained by using the domain values instead of the occurrence of \texttt{NULL}, a strongly possible world is obtained by using the active domain values. 

\begin{definition} A possible world $T^\prime$ of $T$ is called a \emph{strongly possible world (spWorld)} if  $t'[A_i]\in VD^T_i$ for all $t'\in T'$ and $A_i\in R$. 
\end{definition}
The concept of \emph{strongly possible world} was introduced in \cite{alattar2019keys}. Strongly possible worlds allow us to define \emph{strongly possible keys (spKeys)} and \emph{strongly possible functional dependencies (spFDs)}.
\begin{definition}\label{spfd-spkey_def} A strongly possible functional dependency, in notation $X\rightarrow_{sp} Y$, holds in table $T$ over schema $R$ if there exists a strongly possible world $T'$ of $T$ such that $T'\models X\rightarrow Y$. That is, for any $t'_1,t'_2\in T'$ $t'_1[X]=t'_2[X]$ implies $t'_1[Y]=t'_2[Y]$.
  The set of attributes $X$ is a strongly possible key, if there exists a strongly possible world $T'$ of $T$ such that $X$ is a key in $T'$, in notation $\spk{X}$. That is, for any $t'_1,t'_2\in T'$ $t'_1[X]=t'_2[X]$ implies $t_1'=t_2'$. 
\end{definition}
If $T=\{t_1,t_2,\ldots ,t_p\}$ and  $T'=\{t'_1,t'_2,\ldots ,t'_p\}$ is an spWorld of it with $t_i\sim_w t_i'$, then $t_i'$ is called an \emph{sp-extension} or in short an \emph{extension} of $t_i$. Let $X\subseteq R$ be a set of attributes and let $t_i\sim_w t_i'$ such that for each $A\in R\colon t_i'[A]\in VD(A)$, then $t_i'[X]$ is an \emph{strongly possible extension of $t_i$ on $X$ (sp-extension)}
\section{Related Work}\label{rltd-wrk}
Kivinen et. al. \cite{kivinen1995approximate} introduced the measure $g_3$ for total tables.  
Giannella et al. \cite{giannella2004approximation} measure the approximate degree of functional dependencies. They developed the IFD approximation measure and compared it with the other two measures: $g_3$ (minimum number of tuples need to be removed so that the dependency holds) and $\tau$ (the probability of a correct guess of an FD satisfaction) introduced in \cite{kivinen1995approximate} and \cite{goodman1979measures} respectively. They developed analytical bounds on the measure differences and compared these measures analysis on five datasets. The authors show that when measures are meant to define the knowledge degree of $X$ determines $Y$ (prediction or classification), then $IFD$ and $\tau$ measures are more appropriate than $g_3$. On the other hand, when measures are meant to define the number of "violating" tuples in an FD, then, $g_3$ measure is more appropriate than $IFD$ and $\tau$.

In \cite{wijsen2019foundations}, Jef Wijsen summarizes and discusses some theoretical developments and concepts in Consistent query answering CQA (when a user queries a database that is inconsistent with respect to a set of constraints). Database repairing was modeled by an acyclic binary relation $\leq_{db}$ on the set of consistent database instances, where $r_1$ $\leq_{db}$ $r_2$ means that $r_1$ is at least as close to $db$ as $r_2$. One possible distance is the number of tuples to be added and/or removed. In addition to that, Bertossi studied the main concepts of database repairs and CQA in \cite{bertossi2019database}, and emphasis on tracing back the origin, motivation, and early developments.
J. Biskup and L. Wiese present and analyze an algorithm called preCQE that is able to correctly compute a solution instance, for a given original database instance, that obeys the formal properties of inference-proofness and distortion minimality of a set of appropriately formed constraints in \cite{biskup2011sound}. 

\section{Approximation of strongly possible integrity constraints}
\begin{definition}\label{aspkey-def}
Attribute set $K$ is an approximate strongly possible key of ratio $a$ in table $T$, in notation $asp_a^- \left\langle K \right\rangle$, if there exists a subset  $S$ of the tuples $T$ such that $T\setminus S$ satisfies $sp \left\langle K \right\rangle$, and $|S|/|T|\le a$. The minimum $a$ such that $asp_a^- \left\langle K \right\rangle$ holds is denoted by $g_3(K)$. 
\end{definition}
The measure $g_3(K)$ has a value between $0$ and $1$, and it is exactly $0$ when $ sp \left\langle K \right\rangle $ holds in $T$, which means we don't need to remove any tuples. For this, we used the $g_3$ measure introduced in \cite{kivinen1995approximate}, to determine the degree to which $ASP$ key is approximate. For example, the $g_3$ measure of $\spk{X}$ on Table \ref{fig:spkey_main} is 0.5, as we are required to remove two out of four tuples to satisfy the key constraint as shown in Table \ref{fig:spkey_rmv}. 

The $g_3$ approximation measure for spKeys was introduced in \cite{alattar2020strongly}. In this section, we introduce a new approximation measure for spKeys. \begin{definition}
Attribute set $K$ is an add-approximate strongly possible key of ratio $b$ in table $T$, in notation $asp_b^+\left\langle K \right\rangle$, if there exists a set of tuples $S$ such that the table $TS$ satisfies $sp \left\langle K \right\rangle$, and $|S|/|T|\le b$. The minimum $b$ such that $asp_b^+\left\langle K \right\rangle$ holds is denoted by $g_5(K)$. 
\end{definition}
The measure $g_5(K)$ represents the approximation which is the ratio of the number of tuples needed to be added over the total number of tuples so that $sp \left\langle K \right\rangle$ holds. 
The value of the measure $g_3(K)$ ranges between $0$ and $1$, and it is exactly $0$ when $ sp \left\langle K \right\rangle $ holds in $T$, which means we do not have to add any tuple.  For example, the $g_5$ measure of $\spk{X}$ on Table \ref{fig:spkey_main} is 0.25, as it is enough to add one tuple to satisfy the key constraint as shown in Table \ref{fig:spkey_add}.
\begin{definition}\label{spfd-approx-rmv}
For the attribute sets $X$ and $Y$, $\sigma:X \rightarrow_{sp} Y$ is a remove-approximate strongly possible functional dependency of ratio $a$ in a table $T$, in notation \\ $T \models \approx_a^- X \rightarrow_{sp} Y$, if there exists a set of tuples $S$ such that the table $T\setminus S \models X \rightarrow_{sp} Y$, and $|S|/|T|\le a$. Then, $g_3(\sigma)$ is the smallest $a$ such that $T \models \approx_a^-\sigma$ holds. 
\end{definition}
\begin{definition}\label{spfd-approx-add}
For the attribute sets $X$ and $Y$, $\sigma:X \rightarrow_{sp} Y$ is an add-approximate strongly possible functional dependency of ratio $b$ in a table $T$, in notation $T \models \approx_b^+X \rightarrow_{sp} Y$, if there exists a set of tuples $S$ such that the table $T\cup S \models X \rightarrow_{sp} Y$, and $|S|/|T|\le b$. Then, $g_5(\sigma)$ is the smallest $b$  such that $T \models \approx_b^+\sigma$ holds. 
\end{definition}

Let $T$ be a table and $U \subseteq T$ be the set of the tuples that we need to remove so that the spKey holds in $T$, i.e, we need to remove $|U|$ tuples, while by adding a tuple with new values, we may make more than one of the tuples in $U$ satisfy the spKey using the new added values for their \nul s. In other words, we may need to add a fewer number of tuples than the number of tuples we need to remove to satisfy an spKey in the same given table. For example, Table \ref{fig:spkey_main} requires removing two tuples to satisfy $sp \left\langle X \right\rangle$, while adding one tuple is enough.

\begin{table}[h]\caption{Incomplete Table to measure $\spk{X}$} \label{fig:spkey_main}
  \begin{center}
\begin{tabular}{ll}
        			\hline
\multicolumn{2}{c}{\textbf{X}} \\
$X_1$              & $X_2$          
\\\hline
$\bot$ & 1 \\
2 & $\bot$ \\
2 & $\bot$ \\
2 & 2
\\ \hline
\end{tabular}
  \end{center}
\end{table}
\begin{table}[h]\caption{The table after removing ($asp_a^-\left\langle X \right\rangle$)}\label{fig:spkey_rmv}
  \begin{center}
  \begin{tabular}{ll}
        			\hline
\multicolumn{2}{c}{\textbf{X}} \\
$X_1$              & $X_2$          
\\\hline
$\bot$ & 1 \\
2 & 2 \\ \hline
\end{tabular}
  \end{center}
  \end{table}
\begin{table}\caption{The table after adding ($asp_b^+\left\langle X \right\rangle$)}\label{fig:spkey_add}
\begin{center}
  \begin{tabular}{ll}
        			\hline
\multicolumn{2}{c}{\textbf{X}} \\
$X_1$              & $X_2$          
\\\hline
$\bot$ & 1 \\
2 & $\bot$ \\
2 & $\bot$ \\
2 & 2 \\
3 & 3 
\\ \hline
\end{tabular}
\end{center}        		
\end{table}                

\subsection{Relation between $g_3$ and $g_5$ measures}
Results together with their proofs of this subsection were reported in the conference volume \cite{AS-foiks2022}, so the proofs are not included here.
The following Proposition is used to prove Proposition~\ref{g3_l_g5}.
\begin{proposition}\label{prop:removenontotal}
Let $T$ be an instance over schema $R$ and let $K\subseteq R$. If the $K$-total part of the table $T$ satisfies the key $sp \left\langle K \right\rangle$, then there exists a minimum set of tuples $U$ to be removed that are all non-$K$-total so that $T\setminus U$ satisfies $sp \left\langle K \right\rangle.$
\end{proposition}
\begin{proposition}\label{g3_l_g5}
For any $K\subseteq R$ with $|K|\geq 2$, we have $g_3(K)\geq g_5(K)$.
\end{proposition}
Apart form the previous inequality, the two measures are totally independent for spKeys.
\begin{theorem}\label{thm:spkbounded}
Let $0\le\frac{p}{q}<1$ be a rational number. Then there exist tables over schema $\{A_1,A_2\}$ with arbitrarily large number of rows, such that $g_3(\{A_1,A_2\})-g_5(\{A_1,A_2\})=\frac{p}{q}$.
\end{theorem}
Unfortunately, the analogue of Proposition~\ref{prop:removenontotal} is not true for spFDs, so the proof of the following theorem is quiet involved.
\begin{theorem} \label{g3_geq_g5}
Let $T$ be a table over schema $R$, $\sigma: X\rightarrow_{sp} Y$ for some $X,Y\subseteq R$. Then $g_3(\sigma) \geq g_5(\sigma)$.
\end{theorem}
Theorem~\ref{thm:pq-spfd} can be proven by a construction similar to the proof of Theorem~\ref{thm:spkbounded}.
\begin{theorem}\label{thm:pq-spfd}
For any rational number $0\le\frac{p}{q}<1$ there exists tables with an arbitrarily large number of rows and bounded number of columns that satisfy $g_3(\sigma)-g_5(\sigma)=\frac{p}{q}$ for $\sigma\colon X\rightarrow_{sp} Y$.
\end{theorem}
\subsection{Complexity problems}
 \begin{definition}The \emph{SPKey problem} is the following. \newline
    \textsf{Input} Table $T$ over schema $R$ and $K\subseteq R$.\newline
    \textsf{Question} Is it true that $T\models \spk{K}$?\newline
    The \emph{SPKeySystem problem} is the following.\newline
     \textsf{Input} Table $T$ over schema $R$ and $\mathcal{K}\subseteq 2^R$.\newline
     \textsf{Question} Is it true that $T\models \spk{\mathcal{K}}$?\newline
     The \emph{SPFD problem} is the following. \newline
     \textsf{Input} Table $T$ over schema $R$ and $X,Y\subseteq R$.\newline
     \textsf{Question} Is it true that $T\models X\rightarrow_{sp} Y$?
 \end{definition}
 The following was shown in \cite{alattar2020functional}.
    \begin{theorem}\label{thm:spk-in-P} SPKey$\in$P, SPkeySystem and SPFD are NP-complete
   \end{theorem}
However, the approximation measures raise new,  interesting algorithmic questions.
 \begin{definition} The \emph{SpKey-g3} problem is the following. \newline
    \textsf{Input} Table $T$ over schema $R$, $K\subseteq R$ and $0\le q<1$.\newline
    \textsf{Question} Is it true that $g_3(K)\le q$ in table $T$?\newline
    The \emph{SpKey-g5} problem is the following. \newline
    \textsf{Input} Table $T$ over schema $R$, $K\subseteq R$ and $0\le q<1$.\newline
    \textsf{Question} Is it true that $g_5(K)\le q$ in table $T$?\newline
  \end{definition}
 \begin{proposition}\label{prop:spk-g5}
   The decision problem SpKey-g5 is in P.
 \end{proposition}
 \begin{proof} Let us assume that tuples $s_i\colon i=1,2,\ldots ,p$ over schema $R$ are such that $T\cup\{s_1,s_2,\ldots s_p\}$ is optimal, so $g_5(K)=\frac{p}{m}$. Then clearly we may replace $s_i$ by $s'_i=(z_i,z_i,\ldots ,z_i)$ for all $ i=1,2,\ldots ,p$ where $z_i$'s are pairwise distinct new values not appearing in the (extended) table $T\cup\{s_1,s_2,\ldots s_p\}$ so that $T\cup\{s'_1,s'_2,\ldots s'_p\}\models \spk{K}$. Thus, if $g_5(K)\le q$ is needed to be checked for a table $T$ of $m$ tuples, one may add $\lfloor q\cdot m\rfloor$ completely new tuples to obtain table $T'$ and check whether $T'\models \spk{K}$ in polynomial time by Theorem~\ref{thm:spk-in-P}.
 \end{proof}
\begin{theorem}\label{thm:spk-gs} Decision problem SpKey-g3 is in P.\end{theorem}
 \begin{proof}
 Consider the relational schema $R$ and $K\subseteq R$. Furthermore, let $T$ be an instance over $R$ with \nul s. 
Let $T'$ be the set of total tuples  $T'=\{t'\in\Pi_{A\in K}VD^T(A)\colon \exists t\in T \text{ such that } t[K]\sim_w t'[K]\}$, furthermore let $G=(T,T';E)$ be the bipartite graph, called the \emph{$K$-extension graph of $T$}, defined by $\{t,t'\}\in E\iff t[K]\sim_w t'[K]$.
Finding a matching of $G$ that covers all the tuples in $T$ (if exists) provides the set of tuples in $T^{\prime}$ to replace  the incomplete tuples in $T$ with, to verify that $K$ is an spKey.

It was shown in \cite{alattar2020strongly} that the $g_3$ approximation measure for strongly possible keys satisfies
$$ g_3(K) = \frac{|T| -  \nu(G)}{|T|}. $$
where $\nu(G)$ denotes the maximum matching size in the $K$-extension graph $G$.
However, the size of $G$ is usually exponential function of the size of the input of the decision problem SpKey-g3, as $T'$ is usually exponentially large.

In order to make our algorithm run in polynomial time we only generate part of $T'$. Let $T=\lbrace t_1 , t_2 \ldots t_m \rbrace$ and  $\ell(t_i)=|\{t^\star\in \Pi_{A\in K}VD^T(A)\colon t^\star \sim_w t_i[K]\}|$. Note that $\ell(t_i)=\prod_{A\colon t_i[A]=\bot}|VD^T(A)|$, hence these values can be calculated by scanning $T$ once and using appropriate search tree data structures to hold values of active domains of each attribute. Sort tuples of $T$ in non-decreasing $\ell(t_i)$ order, i.e. assume that $\ell(t_1)\le \ell(t_2)\le\ldots \le \ell(t_m)$. Let $j=\max\{i\colon \ell(t_i)<i\}$ and $T_j=\{ t_1,t_2,\ldots t_j\}$, furthermore $T_j^\star= \lbrace t^\star : \exists t \in T_j : t^\star \sim_w t[K] \rbrace  \subseteq \Pi_{A\in K}VD^T(A) $. Note that $|T_j^\star|\le \frac{1}{2}j(j-1) $. If $\forall i=1,2,\ldots ,m\colon \ell(t_i)\ge i$, then define $j=0$ and $T_j^\star=\emptyset$. Let $G^{\star}=(T_j,T_j^\star;E^\star)$ be the induced subgraph of $G$ on the vertex set $T_j\cup T_j^\star$. Note that $|T_j^\star|\le \frac{1}{2}j(j-1) $.\newline
\textit{Claim} $\nu(G)=\nu(G^{\star})+|T\setminus T_j|$. \newline
\textit{Proof of Claim:} 
The inequality $\nu(G)\le\nu(G^{\star})+|T\setminus T_j|$ is straightforward.  On the other hand, a matching of size $\nu(G^{\star})$ in $G^{\star}$ can greedily be extended to the vertices in $|T\setminus T_j|$, as $t_i\in T\setminus T_j$ has at least $i$ neighbours (which can be generated in polynomial time).

Thus it is enough to determine $\nu(G^{\star})$ in order to calculate $g_3(K)$, and that can be done in polynomial time using Augmenting Path method \cite{laszlo}.
 \end{proof}
 Note that the proof above shows that the \emph{exact value} of $g_3(K)$ can be determined in polynomial time. This gives the following corollary.
 \begin{definition}
     The decision problem \emph{SpKey-g3-equal-g5} is defined as
     \textsf{Input} Table $T$ over schema $R$, $K\subseteq R$.\newline
     \textsf{Question} Is $g_3(K)=g_5(K)$?
 \end{definition}
 \begin{corollary}
   The decision problem SpKey-g3-equal-g5 is in P.
 \end{corollary}
 \paragraph{Example}
  Let $R=\{A_1,A_2,A_3\}$, $K_1=\{A_1,A_2\}, K_2=\{A_2,A_3\}$.
  \[T=\begin{array}{lccc}
        &A_1&A_2&A_3\\ \cline{2-4}
        t_1 &1&\bot&1\\
        t_2&1&2&2\\
        t_3&2&1&1\\
        t_4&2&1&1\\\cline{2-4}\end{array}\]
  $T\setminus\{t_4\}\models \spk{K_1}$ and   $T\setminus\{t_4\}\models \spk{K_2}$, but the spWorlds are different. In particular, this implies that for $\mathcal{K}=\{K_1,K_2\}$ we have $g_3(\mathcal{K})>\max\{g_3(K)\colon K\in\mathcal{K}\}$ On the other hand, trivially $g_3(\mathcal{K})\ge\max\{g_3(K)\colon K\in\mathcal{K}\}$ holds. This motivates the following definition.
  \begin{definition}\label{def:max-g3}  
    The problem \emph{Max-g3} defined as \newline
    \textsf{Input} Table $T$ over schema $R$, $\mathcal{K}\subseteq 2^R$.\newline
    \textsf{Question} Is $g_3(\mathcal{K})=\max\{g_3(K)\colon K\in\mathcal{K}\}$?
  \end{definition}
  \begin{theorem} Let Table $T$ over schema $R$ and $\mathcal{K}\subseteq 2^R$. 
 The decision problem Max-g3 is NP-complete.
    \end{theorem}
    \begin{proof}
      The problem is in NP, a witness consists of a set of tuples $U$ to be removed, an index $j\colon \frac{|U|}{|T|}=g_3(K_j)$, also an spWorld $T'$ of $T\setminus U$ such that each $K_i$ is a key in $T'$. Verifying the witness can be done in three steps.
      \begin{enumerate}
      \item $g_3(K_j)\not\le\frac{|U|-1}{|T|}$ is checked in polynomial time using Theorem~\ref{thm:spk-gs}.
      \item For all $i\not= j$ check that $g_3(K_i)\le  \frac{|U|}{|T|}$ using again Theorem~\ref{thm:spk-gs}.
      \item  Using standard database algorithms check that $\forall i\colon K_i$ is a key in $T'$.
      \end{enumerate}
      On the other hand, the SPKeySystem problem can be Karp-reduced to the present question as follows. First check for each $K_i\in\mathcal{K}$ separately whether $\spk{K_i}$ holds, this can be done in polynomial time. If $\forall i\colon T\models \spk{K_i}$ then give $\mathcal{K}$ and $T$ as input for Max-g3. It will answer Yes iff $T\models\spk{\mathcal{K}}$. However, if $\exists i\colon T\not\models \spk{K_i}$, then give the example above as input for Max-g3. Clearly both problems have No answer.
    \end{proof}
    According to Theorem~\ref{thm:spk-in-P}, it is NP-complete to decide whether a given SpFD holds in a table. Here we show that approximations are also hard.
    \begin{definition}
      The \emph{SPFD-g3 (SPFD-g5)} problems are defined as follows.\newline
      \textsf{Input} A table $T$ over schema $R$, $X,Y\subseteq R$, and positive rational number $q$. \newline
      \textsf{Question} Is $g_3(X\spfd Y)\le q$? ( $g_5(X\spfd Y)\le q$?)
    \end{definition}
    \begin{theorem}
      Both decision problems SPFD-g3 and SPFD-g5 are NP-complete.
    \end{theorem}
    \begin{proof}
      To show that SPFD-g3$\in$NP one may take a witness consisting of a subset $U\subset T$, an spWorld $T^\star$ of $T\setminus U$ such that $T^\star\models X\rightarrow Y$ and $|U|/|T|\le q$. The validity of the witness can easily be checked in polynomial time. Similarly, to show that SPFD-g5$\in$NP one may take a set of tuples $S$ over $R$ and an spWorld $T^\star$ of$T\cup S$ such that $T^\star\models X\rightarrow Y$ and $|S|/|T|\le q$.

      On the other hand, if $|T|=m$ and $q<1/m$, then both SPFD-g3 and SPFD-g5 are equivalent with the original SPFD problem, since the smallest non-zero approximation measure is obtained if one tuple is needed to be deleted or added. According to Theorem~\ref{thm:spk-in-P}, SPFD problem is NP-complete, thus so are SPFD-g3 and SPFD-g5.
      \end{proof}
 \subsection*{Acknowledgement}
  Research of the second author was partially supported by the
    National Research, Development and Innovation Office (NKFIH)
    grants K--116769 and SNN-135643. This work was also   supported by the BME-
Artificial Intelligence FIKP grant of EMMI (BME FIKP-MI/SC) and by the Ministry of Innovation and
Technology and the National Research, Development and Innovation
Office within the Artificial Intelligence National Laboratory of Hungary.
   
\bibliographystyle{abbrv}
\bibliography{spFD-key_approx} 
\end{document}